\documentclass[a4paper,9pt]{amsart}%,reqno
\usepackage{multicol}

\usepackage[a4paper]{geometry}
\usepackage{amsthm,hyperref,lmodern,mathrsfs}
\usepackage[utf8]{inputenc}
\usepackage[T1]{fontenc} 
\usepackage{graphicx}
\usepackage[english]{babel}
\usepackage{calc}
\usepackage{amsmath,amsthm}

\usepackage{url}
\usepackage{times, amssymb, amscd, mathrsfs, graphicx, color,soul}  
\usepackage{enumitem}%{enumerate}

\usepackage{tcolorbox}
\usepackage[tikz]{bclogo}

\usepackage{fancybox}

\newcommand*{\colorboxed}{}
\def\colorboxed#1#{\colorboxedAux{#1}}
\newcommand*{\colorboxedAux}[3]{%
  \begingroup
    \colorlet{cb@saved}{.}%
    \color#1{#2}%
    \boxed{%
      \color{cb@saved}%
      #3%
    }%
  \endgroup
}

\usepackage[indentafter]{titlesec}
\titleformat{name=\section}[runin]{}{\thetitle.}{0.5em}{\bfseries}
\titleformat{name=\subsection}[runin]{}{\thetitle.}{0.5em}{\itshape}

\newtheorem{theo}{Theorem}[section]

\newtheorem{lem}[theo]{Lemma}

\definecolor{darkred}{rgb}{0.9,0.1,0.1}

\title{Galton-Watson process and Bayesian inference: \\
A turnkey method for the viability study of small populations}
\author{B. Cloez, T. Daufresne, M. Kerioui, B. Fontez}

\date{ Compiled \today}

\begin{document}
\maketitle

\begin{abstract} \ \\
\textbf{1} Sharp prediction of extinction times is needed in biodiversity monitoring and conservation management.\\
\textbf{2} The Galton-Watson process is a classical stochastic model for describing population dynamics. Its evolution is like the matrix population model where offspring numbers are random. Extinction probability, extinction time, abundance are well known and given by explicit formulas. In contrast with the deterministic model, it can be applied to small populations.\\
\textbf{3} Parameters of this model can be estimated through the Bayesian inference framework. This enables to consider non-arbitrary scenarios.\\
\textbf{4} We show how coupling Bayesian inference with the Galton-Watson model provides several features: i) a flexible modelling approach with easily understandable parameters ii) compatibility with the classical matrix population model (Leslie type model) iii) A non-computational approach which then leads to more information with less computing iv) a non-arbitrary choice for scenarios, parameters... It can be seen to go one step further than the classical matrix population model for the viability problem.
\\
 \textbf{5} To illustrate these features, we provide analysis details for two examples whose one of which is a real life example. 
 \end{abstract}

\begin{multicols}{2}

\section{Introduction}

Population viability analysis (PVA) \cite{B84,S87,B92,P00,MD02} aims to manage or predict population development from some ecological, genetic or demographic assumptions. It generally relies on modelling purposes. The simulations are frequently used to work on these questions even if some works highlight the use of Mathematical results for these problems as \cite{Caswell,GL00,LGN07,OM10,ELSW} among many others.

Habitat degradation is pointed out as the major cause of extinction for many threatened species around the world. Yet, for species that are exploited for food or other purposes (e.g., exploited fish populations, bushmeat etc...), or purposely destroyed (e. g., large predators), extinction often occurs in intact habitats. For these cases, the assumption that population dynamics takes place in a fixed habitat seems reasonable. Then, assuming a given environment and a  short period of time without brutal variation of the environment, are we able to propose a model to support decision making for management? 
 
In this paper, we will illustrate how to answer some questions emanating from conservation biology such as:
\begin{itemize}
\item[(I)] What are the odds on saving a population? How probable is extinction? For how long can the population survive? 
\item[(II)] How sharply are numbers expected to fall in the short term? 
\item[(III)] How can we proceed to avoid extinction \cite{CFB,CQLC}, to maintain a population size \cite{SGRFBG} or wipe out some population \cite{BACS, LBBD,GKMMRS}. For instance, if we are able to add some individuals in a population, how many re-introductions are needed to reach a threshold of variability for the population with high probability?
%\item If the population is not viable, how many times it will survive?
\end{itemize}

These simple questions, when the environment is assumed unchanged/stable are not currently being resolved successfully and we propose a reformulation of classical models that leads to understand the population viability with or without possible actions.

To illustrate this, Cairns, Ross and Taimre in \cite{CRT} underlined that the usual class of diffusion models are simple but lead to inaccurate predictions of critical values
such as the expected time to extinction. They advocated that a more appropriate model may be a discrete-state Markov process
describing the actual number of individuals in the population like the birth–death processes. These models are not used in practice because they are  difficult to work with, from both analytical and computational points of view.

Starting from the famous Matrix population \cite{B41,L42,L45}, we will  introduce the simplest birth-death process of Galton Watson which leads asymptotically to exponential growth and does not take into account some density-dependence dynamics. This assumption is relevant if we address population dynamics in populations at low density, whose dynamics is generally not subject to density dependence factors. On the other hand, these populations may exhibit small numbers, which may preclude the use of deterministic models. We will show step by step how to use it for extinction prediction.

Although extensive literature exists on the frequentist estimation for the Galton-Watson processes (see for instance \cite{AK78,C97,D93,MT05} and references therein), we choose a Bayesian reasoning \cite{R07}. In contrast with a frequentist approach, it has the advantage that a prior with expert knowledge or previous analyses can be introduced in the models. This may be especially useful for small data sets. Indeed, data sets for rare and elusive species, which are often the target of PVA, are generally scarce and incomplete. Moreover, a posterior allows a precise quantification of the uncertainty due to the parameter estimations. This gives posterior probability estimates or credible intervals for the various parameters allowing to build natural scenarios without arbitrary choice. This is of particular interest for practitioners.

Besides, our approach is based on a remarkable relationship between the Galton-Watson processes and Bayesian inference. As we will see, the Galton-Watson processes have multinomial transitions and verify the Markov property, the law of the parameters conditioned according to the historical paths (namely the posterior) can have an explicit form. Indeed, choosing a classical conjugate prior of the multinomial distribution enables recovery of this distribution as a posterior (with different parameters) after one generation and then the Markov property enables to generalize this property for all generations.

This allows to use classical formulas on Galton-Watson processes and integrate it under the posterior to predict quantitatively the behaviour of the population; that is the growth, extinction, relative abundance...

This simple and powerful expression of the posterior is one of the main strengths of our approach. Indeed, this remarkable property seems to be observed only in \cite{MG00} in the mono-type case (i.e. $K=1$), and, to the best of our knowledge, has never been applied in practice with a mono-type nor with many types. We will show how this leads to a simple and powerful way to analyse the viability of the population.
 
\textbf{Outline:} In the two next Sections, Section~\ref{sect:GWintro} and Section~\ref{sect:BGW}, we present the Galton-Watson processes and our stastical model. Section~\ref{sect:BGW} is more technical and contains several formulas and properties of such processes. Section~\ref{sect:stat} states the posterior on which we will apply the properties developed in Section~\ref{sect:BGW}. Section~\ref{sect:exemple} is devoted to two pedagogical examples. the first is based on synthetic data and is concerned with extinction. The second concerns the French bear population and survival, which apart from illustrating our method, is of great interest in is own right. Finally Section~\ref{sect:ext} gives some useful simple extensions.
 
 \begin{tcolorbox}[colframe=white, colback=blue!20]

\textbf{Box 0.} Notation
  \hrule %\hrulefill 
\

\begin{itemize}[leftmargin=+0.2cm]
\item $K$ is the number of types.
\item $N(t)= (N_1(t), \dots ,N_K(t))^\intercal$ designes the state vectors (sample size of each type) at time $t$.
\item $p_{i,j}(k)$ is the probability that an individual of type $i$ gives $k$ individuals of type $j$ after one unit of time. It is supposed to be a random number.
\item $\kappa_{i,j}$ is the maximal number of offspring of type $j$ from an individual of type $i$.
\item $(\alpha_{i,j}(k))$ design the prior knowledge in terms of sample size of individuals in type $i$ which gives $k$ individuals of type $j$.
\item $\xi_{i,j,l,t}$ is the (random) number of offspring of type $j$ of the $l$-th individual of type $j$.
\item $M= (M_{i,j})= \left(\sum_{k= 0}^{\kappa_{i,j}} k p_{i,j}(k)\right)$ is the random \textit{mean} matrix.
\item $\lambda$ is the principal eigenvalue of $M$.
\item $u,v$ are respectively the left and right eigenvector of $M$.
\item $T_{\text{Ext}}$ is the (random) extinction time. Namely the duration after the present time $T$ to have extinction.
\item $s_i$ is the probability of extinction starting from one individual with type $i$.
\item $\mathbf{n}=(n_{i,j}(k,t))_{1\leq i,j\leq K, k \geq 0, T\geq t\geq 0}$ are the data. they correspond to the number of individuals of type $i$ with $k$ offspring of type $j$ at time $t$ for all types and times over $\{0,...,T\}$. 
\item $\alpha$ is a statistical threshold. Classicaly $\alpha=0.05$ or $0.1$.
\end{itemize}

\end{tcolorbox}

\begin{tcolorbox}[colframe=white, colback=yellow!20]

\textbf{Box 1.} Pathways to PVA through Bayesian-Galton-Watson approach
  \hrule \  %\hrulefill 
  
Here, we summarize some steps for carrying out a viability study using data. Several details are given in Appendix~\ref{sect:box1expanded}.
\begin{enumerate}[leftmargin=+0.2cm]
\item  \textbf{Previous knowledge}: Initialize your hyper parameters $(\alpha_{i,j}(k))$.
%\begin{itemize}[leftmargin=-0.1cm]
%\item 
%\end{itemize}
%\item Estimation:
%One has
%$$
%p_{i,j} \in []
%$$
%with ...
%\item Prediction:
%\begin{itemize}[leftmargin=-0.1cm]
%\item Extinction: 
%\end{itemize}
\item Compute the following quantity of interest:
\begin{itemize}[leftmargin=-0.1cm]
\item \textbf{Short time evolution}: number of individuals:
$$\mathbb{E}[ N(t) \ | \ \mathbf{n}] = \mathbb{E}[M^t \cdot N(0)  \ | \ \mathbf{n}].$$
\item \textbf{Probability of viability}:  
$$\mathbb{P}(\lambda>1 \ | \ \mathbf{n}).$$
If $\lambda<1$ then reintroduction is useless, one has to act on the environment.
\item \textbf{Time to extinction}: if $\lambda <1$, extinction is certain and one can draw
$$
U(t): t \mapsto \mathbb{E}[\lambda^t  \ | \ \mathbf{n} ] \times \frac{\sum_{j=1}^K N_j(0) v_j}{ \min_{1 \leq j \leq K} v_j}.
$$
Indeed when $U(t)$ is below a threshold $\alpha$ (as $1$ or $5 \%$) we know that after this time there is more than $1-\alpha$ that all individuals are extinct; that is $T_{\text{Ext}} \leq t$. An upper bound for $T_{\text{Ext}}$ is given in Section~\ref{sect:box1expanded}.
\item \textbf{Probability of extinction}: if $\lambda>1$ then the population has the capacity to survive but  extinction can occur if the population is too small:
$$
\mathbb{P}(\text{Ext})= \mathbb{E}\left[ s_1^{N_1(0)} \times \cdots \times s_K^{N_K(0)} \ | \ \mathbf{n}  \right]
$$
\item \textbf{How to plan reintroduction}: draw the laws of $s_1$,$\dots$, $s_K$ to see the survival potential of each type and therefore choose to reintroduce the more efficient one.
\end{itemize}

\end{enumerate}
%\textbf{Box1.} PVA method via Bayesian-Galton-Watson approach

\end{tcolorbox} 

\section{Galton Watson process: a Stochastic Matrix Population Model}
\label{sect:GWintro}
 One of the most famous and powerful mathematical settings for PVA is the Matrix population model. They were introduced by \cite{B41,L42,L45} and they have become a classical tool to predict the evolving size of a population as illustrated in the books \cite{Caswell, TC97}. In this model, a population is classified into discrete types
and its evolution is in discrete time. The state of the population, at a certain time $t\in \mathbb{N}=\{0,1,...\}$, is given
by a (column) vector, $N(t) = (N_1(t), \dots ,N_K(t))^\intercal$, where $K$ is the number of types (and the exponent $\intercal$ designs the usual transposed vector) . This vector evolves from time $t$ to $t + 1$ with the following matrix product:
$$
N(t+1)= A \cdot N(t),
$$
that is, for every type $j$,
\begin{equation}
\label{eq:mpm}
N_{j}(t+1) = \sum_{i=1}^K A_{i,j} N_{i}(t).
\end{equation}
Matrix $A$ is called the population projection matrix and the entry $A_{i,j}$ tells how many individuals of type $i$ appear per individual of type $j$.
Classically the types correspond to the age (or the location), and $A$ then codes the development (or the migration), the birth and the death. %Even if $A$ is often supposed to be constant over time, when matrix $A$ changes at each time in a stochastic fashion, this model is called stochastic matrix population model [Caswell, chap? ]. This adds randomness in the environment but not in the population dynamics. In what follow, we do not consider this case.

This classical model makes it possible to understand population variation and resulting extinction or viability. Indeed, let us consider for instance the one-type case $K=1$. We then have
\begin{equation}
\label{eq:K1}
N_1 (t+1) = A_{1,1} N_1(t) \Rightarrow N(t) = A_{1,1}^t N_1(0),
\end{equation}
for all $t\geq 0$, which entails
$$
\lim_{t\to \infty} N_1(t) = \left\{
    \begin{array}{ll}
        0 & \text{if } A_{1,1} <1 \\
        N_1(0) & \mbox{if } A_{1,1} =1 \\
        + \infty & \mbox{if } A_{1,1} >1.
    \end{array}
\right.
$$
In particular when $A_{1,1}>1$ the population is viable and when $A_{1,1}<1$ the population goes into extinction. This result can easily be generalized with $K$ types using the principal eigenvalue of $A$. Namely if $\lambda$ is the largest eigenvalue of $A$ (in modulus) and $N$ the associated eigenvector then $N(t) \underset{t \to \infty}{\sim} \lambda^t N$. The vector $N$ then relates the asymptotic relative abundance although $\lambda$ represents the growth rate. See for instance \cite{Caswell} for details.

In case of extinction, a central question in PVA is to have an estimate of the extinction time. Nevertheless, any prediction of the extinction time is impossible since population size is a deterministic real number. Namely, we have $N_1(t)>0$ for all $t\geq 0$ as soon as $N_1(0)>0$. When $K=1$ and $A_{1,1}<1$, or more generally $\lambda<1$, extinction means $\lim_{t \to \infty} N(t)=0$ but at each time there are some individuals even if this number can be very small. Hence, mathematically the extinction time cannot be defined as the time $T_{\text{Ext}}$ that $N_1(T_{\text{Ext}})=0$ for the matrix population model. The usual setting to quantify time to extinction in this setting is therefore the so-called quasi-extinction time \cite{GSJB,HSVW,ML91,DACMMM}. That is, we fix some arbitrary threshold $\gamma$ under which the number of individuals will be insufficient to ensure persistence of the population, and the associated quasi-extinction time $T_\gamma$ is defined as the moment the population size reaches the level $\gamma$. This threshold has to be chosen in order to take into account demographic stochasticity, Allee effect, \textit{etc}. When $K=1$, Equation~\eqref{eq:K1} gives
$$
T_\gamma= \frac{\log(\gamma) - \log(N_1(0))}{\log(A_{1,1})},
$$
and we see that the problem is that this time crucially depends on the arbitrary choice of $\gamma$.

Another drawback to considering real numbers to model the population size is the unit choice. Indeed does $N_1(t)=1$ mean that there is $1, 10, 10^2, 10^3,...$ individuals? Then how can we use this type of model for small population sizes as in \cite{OURS16,FF10} or \cite[Section 5]{ELSW} for instance? Note that, as deterministic models are proved to be realistic only in large populations, this unit choice has to be large. But how many individuals is enough to be large depends on the model and its parameters. Even if deterministic models can give suitable results for only a hundred individuals, a unit choice of $10^3$ (or even $10^6$) may not be enough to avoid atto-fox type problem \cite{CL12} (that is having  $10^{-18}$ individuals able to avoid extinction).

As pointed out in \cite{GL00,LGN07}, to overcome these problems, it is more natural to consider the population size as a discrete number (namely an integer) as this is the case in Galton-Watson processes.

Multi-type Galton-Watson processes, were introduced in \cite{B1845,GW}, and can be seen as a matrix population model with offspring randomness (but with non-random environment). In this approach, at each time $t$, each individual $l$ of type $i$ gives a discrete random number $\xi_{i,j,l,t} \in \{0, 1, 2, \dots \}$ of child type $j$; namely formula \eqref{eq:mpm} became
\begin{equation}
\label{eq:GW}
N_{j}(t+1) = \sum_{i=1}^K \sum_{l=1}^{N_{i}(t)} \xi_{i,j,l,t}.
\end{equation}
Of course setting 
$$
A_{i,j}(t)= \frac{1}{N_{i}(t)} \sum_{l=1}^{N_{i}(t)} \xi_{i,j,l,t},
$$
we recover the form of equation~\eqref{eq:mpm} that is 
$$
N(t+1)=A(t) \cdot N(t).
$$
With this formalism, $A(\cdot)$ is a random matrix changing at each time. However it is not the same as in \cite[Chapter 3]{TC97} because this sequence of matrices is not a sequence of independent and identically distributed matrices. To see how demography or environment play a role in this matrix notation and to see the importance of considering demographic stochasticity (even in larger populations), we can see \cite{ELSW} and references therein. 

Consistently, as in matrix populations, one can recover the asymptotic type distributions and the increasing rate of the population through decomposition of a matrix, but one can go further and find sharp results on extinction.  The credit that can be given to these predictions can also be estimated. That is, we can give a precise confidence interval for the extinction risk, the extinction time, $etc.$ This again is one of the main differences with matrix population models. The proof of these results can be found in \cite{AN04,KA02,HJV,H02,M16}. Some of them have already been applied for real-life examples; see for instance \cite{CFB} or \cite[page 208]{GL00}. However, besides using the mathematical properties of Galton-Watson processes, we also propose in this work an efficient and consistent statistical approach to using it. Indeed, as noticed in \cite{OM10} "Conducting a population viability analysis involves the steps of choosing an appropriate model, fitting the model to data, and using the fitted model to predict the extinction risk ", the approach that we describe aims to cover all these PVA steps from some life tables or more generally from demographic data.

To conclude, one should keep in mind, that as with matrix population models, the Galton-Watson process crucially depends on the two following strong assumptions:
\begin{enumerate}
\item  interaction between individuals has no effect on the population dynamics;
\item parameters of one individual depend on this individual only through its type. In particular, all individuals of the same type have the same distribution of offspring and this does not vary in time.
\end{enumerate}

\iffalse

\begin{bclogo}[couleur = blue!5,sousTitre=Notation, marge =10, ]{ Box 0.} 
contenu
\end{bclogo}

\boxput*(0,1){
\colorbox {white}{Titre de la boite}
}{
\setlength{\fboxsep }{6pt}
\fbox{ \begin{minipage}{8cm}
Bla \\
Bla \\
Blaadditionally
\end{minipage}}
}
\fi

\iffalse

\begin{bclogo}[couleur = yellow!5,sousTitre= PVA method via Bayesian-Galton-Watson approach, marge =10, ]{ Box 1.} 
contenu
\end{bclogo}

\boxput*(0,1){
\colorbox {white}{Titre de la boite}
}{
\setlength{\fboxsep }{6pt}
\fbox{ \begin{minipage}{8cm}
Bla \\
Bla \\
Bla
\end{minipage}}
}
\fi

\section{Our model: a Galton-Watson process embedded in a Bayesian framework}
\label{sect:BGW}

To define a multi-type Galton-Watson process, we need to introduce some probabilities $p_{i,j}$ on $\{0,1,...\}$. They will represent the progeny of type $j$ from an individual of type $i$. More precisely, an individual oftype $i$ will give, after one unit of time, $k$ individuals of type $j$, with probability $p_{i,j}(k)$. We assume the existence of a maximal number of offspring $\kappa_{i,j}$; that is $p_{i,j}(k)=0$, for all $k> \kappa_{i,j}$. We have then $K^2$ probabilities $p_{1,1},p_{1,2},..., p_{K,K}$ which are represented in a vector form $p_{i,j}= (p_{i,j}(0),..., p_{i,j}(\kappa_{i,j}))$.

For real life problems, these (fundamental) parameters $p_{i,j}$ are unknown. In a frequentist framework,
these parameters of interest would have been assumed to be unknown,
but fixed. Namely, it would be assumed that in the population, for all $i,j$
there is only one true probability distribution $p_{i,j}$. In the Bayesian view of subjective probability, all of these unknowns are treated as uncertain and therefore should be described by a probability distribution over probability distribution.

We will then suppose that $(p_{i,j})_{1\leq i,j \leq K}$ is distributed according to some probability distribution, usually called the prior. All through this paper, this prior will be the convolution of Dirichlet distribution \cite[Section  A.8 p. 521]{R07}; namely $p_{i,j}$ are independent from each other and probability-valued random variables whose law is given for every $i,j$ by
$$
p_{i,j} \sim \mathbf{Dir}\left(\alpha_{i,j}(0),..., \alpha_{i,j}(\kappa_{i,j})\right),
$$
for some sequence $(\alpha_{i,j}(k))_{k=0,...,\kappa_{i,j}}$. This is a classical law on discrete probability whose  expression is given in Appendix \ref{sect:post}. Parameters $(\alpha_{i,j})$ are called hyper parameters and correspond to the information that you want to incorporate in your estimation. They are fixed by the users and can take into account previous studies or expert knowledge.

Among all the choices, two main strategies can be chosen. 

The first one corresponds to the non-informative choice. If we have no information on the ecological system, we wish to decrease the effects of prior outcomes. We then take $p_{i,j}$ as drawn uniformly at random over all probability vectors. This corresponds to the choice $\alpha_{i,j}(0)=\dots=\alpha_{i,j}(\kappa_{i,j})=1$.
Namely, we consider as prior, the uniform law on the simplex of probability measure on a discrete space. This choice is motivated by the minimization of the entropy of these distributions (See \cite[Section 3.2.3]{R07} and \cite[Section 6]{R09}). Roughly, it is the more random choice or the less informative one.

Another choice can consist of putting information in the hyperparameters $(\alpha_{i,j})$. Indeed, as Dirichlet distribution is uni-modal (at least for large parameters) then one can easily choose $(\alpha_{i,j})$ so that $\mathbb{E}[p_{i,j}(k)]$ is equal to a presupposed valued $m_k$ plus or minus a fixed error estimate $\sigma$ (uniform in $k$). Values of $m_k$ and $\sigma$ can for instance be taken from a previous study where data are not available but just the estimator with a confidence interval as in \cite{CQLC,CFB}. 
Then the expert choice is driven by external data and the weight of the expertize is automatic. Finally $(\alpha_{i,j})$ can incorporate belief which does not come from any data. These parameters also represent  the balance between data and expertize.

Details for initializing $(\alpha_{i,j})$ in practice are given in Appendix~\ref{sect:box1expanded}.

Forthwith,  a Galton-Watson process $N= (N(t))_{t=1,2,...}$ is a random vector-valued sequence satisfying equation \eqref{eq:GW}. In this equation, conditionally on $(p_{i,j})$, random variables $\xi_{i,j,l,t}$ are supposed to be independent and $\xi_{i,j,l,t}$ are distributed according to $p_{i,j}$ on $\{0,1,..., \kappa_{i,j} \}$. Namely, 
\begin{align*}
\mathbb{P}&\big(\forall i,j,l,t, \  \xi_{i,j,l,t} \in A_{i,j,l,t} \ | \ (p_{i,j}) \big) \\
&= \prod_{l,t\in \mathbb{N}} \prod_{i,j=1}^K \sum_{k \in  A_{i,j,l,t} } p_{i,j}(k),
\end{align*}
for all sets $A_{i,j,l,t} \subset \mathbb{N}$. This is a discrete-time model. Although there is a continuous-time version of the Galton-Watson processes and our result should hold for continuous times, we restrict ourself to a discrete-time setting because of the general nature of the data. Indeed data is often collected at punctual moments. Number $N(t)$ then represents the number of individuals at  $t$ years,  $t$ days or more relevant choices of times.

An important quantity associated with this model the  random \textit{mean} matrix $M$ whose entries are
\begin{equation}
\label{eq:M}
M_{i,j}= \sum_{k= 0}^{\kappa_{i,j}} k p_{i,j}(k).
\end{equation}

It is a random matrix since $p_{i,j}$ are supposed random (Bayesian framework). This matrix is the counterpart to the population projection matrix in matrix population model \cite{Caswell, TC97}. When $p_{i,j}(0)=1$ for $j\notin \{0,i,i+1\}$ that is one can just get old or reproduce then $M_{i,j}=0$ for $j\notin \{0,i,i+1\}$ and we recover the classical Leslie matrix. Under general assumptions \cite[Section 2.3.1]{HJV}, the evolution of $N(t)$ when $t$ becomes large only depends on the largest eigenvalue $\lambda$ of this matrix and the two associated eigenvectors $u,v$ (with positive coordinates). That are those verifying
\begin{equation}
\label{eq:vp}
uM= \lambda M, \ M v = \lambda v
\end{equation}
and
$$
\sum_{i=1}^K u_i v_i =1, \ \sum_{i=1}^K u_i =1
$$

On the event $\lambda<1$ then the process goes to extinction; that is there exists a finite random time $T_{\text{Ext}}$ such that $N(t)=0$ for all $t\geq T_{\text{Ext}}$ (and $N(T_{\text{Ext}}-1)>0$). Even if one can estimate $T_{\text{Ext}}$ by simulation, sharp analytic bounds have been proved. For instance,
\begin{align}
&\mathbb{P}(T_{\text{Ext}} >t \ | \  N_1(0), \dots , N_K(0), (p_{i,j)} )\label{eq:tps-ext}\\
=\ &\mathbb{P}\left(\sum_{i=1}^K N_i(t)>0 \ | \  N_1(0), \dots , N_K(0), (p_{i,j)} \right)\nonumber \\
\leq \ &\lambda^t \frac{\sum_{i=1}^K v_i N_i(0) }{ \min_{1\leq i \leq K} v_i}.\nonumber
\end{align}
See for instance \cite[Equation (5.54)]{HJV}. This bound allows to find an upper bound of the extinction time with high probability (see Box 1 and Section~\ref{sect:box1expanded}). To our knowledge, there was no equivalent lower bound in great generality. However, \cite[Box 5.2 p. 119]{HJV} gives such a bound when $K=1$. This allows to have a lower bound on the extinction time. In Appendix~\ref{sect:box1expanded}, we generalize this result for $K>1$ and show how use it and \eqref{eq:tps-ext} to bound the extinction time.

On the event $\lambda>1$, extinction occurs with positive probability (but not equal to $1$). In fact, if Ext denotes the extinction event then we have
\begin{align}
\label{eq:proba-ext}
\mathbb{P}&(\text{Ext}  \ | \ N_1(0), \dots , N_K(0), (p_{i,j)})\\ 
&= s_1^{N_1(0)} \times \cdots \times s_K^{N_K(0)},\nonumber
\end{align}
where $s=(s_1,...,s_K)$ is the unique solution of $s=\varphi(s)$ in the simplex $\triangle=\{s \in [0,1]^K \ | \ \sum_{i=1}^K s_i=1 \}$. Function $\varphi= (\varphi_i)_{1\leq i \leq K}$ is the generating function associated with $p$; for $i\in \{1,...,K\}$, it is defined by
\begin{equation}
\label{eq:phi}
\varphi_i(s_1,...,s_K) = \sum_{i_1,...,i_k} s_1^{i_1}p_{i,1}^{i_1} \dots s_K^{i_K} p_{i,K}^{i_K}.
\end{equation}
In particular, when we start with only one individual with type $i$ (that is $N_i(0)=1$ and $N_j(0)= 0$ for $j\neq i$) then
\begin{align*}
\mathbb{P}&(\text{Ext} \ | \ N_1(0), \dots , N_K(0), (p_{i,j)})\\ 
&= s_i.
\end{align*}

Additionally, Process $(N(t))$ is a Markov process. That means that for a current time $T$, the future evolution of the population only depends on the distribution of the population at time $T$ and not on periods of time in the past (conditionally on $(p_{i,j})$ of course). As a consequence, Formula~\eqref{eq:tps-ext} Formula~\eqref{eq:proba-ext} hold when replacing $N(0)$ by $N(T)$ and conditioning on $N(T)>0$. Know that some individuals remain at time $T$, we have to take into account this information.

Finally, in addition to the spectral analysis, all results of a matrix population model (sensitivity analysis,...) can be applied using the matrix $M$ for the study of the Galton-Watson process. Also, for short time, the Galton-Watson process evolves in the same way as this associated deterministic model up to error around $\sqrt{N(t)}$. 

All of these mathematical results can be found in \cite{AN04,KA02,HJV,H02,M16}.

\section{Main result: explicit expression of the posterior}
\label{sect:stat}

We assume that our demographic data has the form of life tables $ (n_{i,j}(k,t))_{1\leq i,j\leq K, k,t\geq 0}$; where $n_{i,j}(k,t)$ corresponds to the number of individuals at time $t$ of type $i$ which gives $k$ individuals of type $j$ at time $t+1$. As for example in Table~\ref{tab:data-sim}. We assume we know the population over the time interval $\{0, \dots, T\}$; namely our data is $\mathbf{n}=(n_{i,j}(k,t))_{1\leq i,j\leq K, k \geq 0, T\geq t\geq 0}$.

Due to the well known conjugation property between multinomial and Dirichlet distribution (see for instance \cite{R07,R09}), and the form of the Galton-Watson transitions, the form of the posterior is simple. More precisely, we can see this for every $i,j$, the posterior of the parameter $p_{i,j}$ conditioned on the data $\mathbf{n}$ are
\begin{align}
\label{eq:post-gen}
&\mathcal{L}(p_{i,j} \ | \ \mathbf{n})\\
=&\mathbf{Dir}\left(\alpha_{i,j}(0) + \sum_{t=0}^T n_{i,j}(0,t),\right. \nonumber\\
&\qquad \qquad  \left. \dots, \alpha_{i,j}(\kappa_{i,j}) + \sum_{t=0}^Tn_{i,j}(\kappa_{i,j},t)\right)\nonumber.
\end{align}
\iffalse
As a consequence, we for instance obtain
\begin{equation}
\label{eq:post-mean}
\mathbb{E}[M_{i,j} \ |\ (n)]= \frac{\sum_{k=1}^{\kappa_{i,j}} k \left( 1 + \sum_{t=0}^Tn_{i,j}(k,t) \right)}{\sum_{k=0}^{\kappa_{i,j}} \left( 1 + \sum_{t=0}^Tn_{i,j}(k,t) \right)}
\end{equation} 
\fi
\iffalse
extensive summary of the
information available on the parameter, integrating simultaneously prior
information and information brought by the observation

This posterior sum up then all the information we have coming from the data and our population model. To answer our main questions, it remains to summarize this object to the principal objects of interest (chance of extinction, extinction time, abundance, ...). According to the objectives of the outcomes, many choices of modelling are over possible and we will describe three different ones: an estimator approach, a scenario approach or a direct integration approach.

The first-one consist to build an estimator 

[Christian P. Robert
The Bayesian Choice 3.2.3]
\fi
The proof of this formula is given in Appendix~\ref{sect:post}.

This posterior turns into an extensive summary of the
information available on the parameter, integrating simultaneously modelling choices and available observations. To answer our main questions, we must summarize this object into the principal objects of interest (risk of extinction, extinction time, abundance, ...). Box 1 and Appendix~\ref{sect:box1expanded} gives a typical "to-do" list in order to answer some main questions in PVA. We directly use the posterior to calculate precisely the probability of some events. For instance, one can calculate
$$
\mathbb{P}(Ext \ | \ (n)),
$$
directly integrating Equation~\eqref{eq:proba-ext} under the Dirichlet distribution given in Equation~\eqref{eq:post-gen}. This has the advantage of not being conditioned on one $(p_{i,j})$ because we take a mean value of all $(p_{i,j})$ combinations, using our posterior. This perfectly describes the prediction randomness coming from the demographic stochasticity and the statistical estimation.% It does not resume the information. 

One possible drawback is that all quantities of interest (abundance...) are random objects and may be difficult to understand or describe. Using the full posterior can then be difficult for practitioners and it can be helpful to sum up the information. We describe here two possible ways of aggregating the posterior according to the objectives of the outcomes: an estimator approach and a scenario approach. 

Firstly, as explained in detail in \cite[Section 4]{R07}, the posterior mean, the posterior mode, the posterior median, $etc.$ can be used to build an estimator of the input of interest. If a choice must be made among the types of estimators above, there is no way
of selecting a best estimator without using an error-type criterion; see  \cite[Section 4]{R07} for details. Note that some of these types of estimator can be easily computed; for instance
\begin{equation}
\label{eq:post-mean}
\mathbb{E}[M_{i,j} \ |\ \mathbf{n} ]= \frac{\sum_{k=1}^{\kappa_{i,j}} k \left( \alpha_{i,j}(0) + \sum_{t=0}^T n_{i,j}(k,t) \right)}{\sum_{k=0}^{\kappa_{i,j}} \left( \alpha_{i,j}(0) + \sum_{t=0}^Tn_{i,j}(k,t) \right)}.
\end{equation} 
The main drawback of this approach is that it loses the error estimation which is naturally embedded in the posterior. However, as soon as these values are fixed one can directly calculate $\lambda$ (estimation of viability), probability of extinction $s_1, s_2$ ...
%(even if one can naturally build non-symmetric confidence interval)

%\begin{tcolorbox}[colframe=white, colback=yellow!20]
%\textbf{Box 2.} Pathways to estimate a popumation projection matrix from posterior mean estimation
%\hrulefill
%%\setlist[itemize]{leftmargin=*}
%\begin{enumerate}[leftmargin=*]
%\item  Previous knowledge:
%\begin{itemize}[leftmargin=-0.1cm]
%\item If you have no information, set 
%$$
%\alpha_{i,j} (k) = 1.
%$$
%\item If you expect that $p_{i,j} = m_{i,j}(k) \pm \epsilon_{i,j}(k)$ and $\epsilon_{i,j}(k)$ is centred and has variance of magnitude $$\sigma^2_{i,j}  m_{i,j}(k) (1-m_{i,j}(k))$$ then set
%$$
%\alpha_{i,j} (k) = \frac{(1 - \sigma^2_{i,j}) m_{i,j}(k)}{\sigma^2_{i,j} \sum_{k=0}^\kappa m_{i,j}(k)},
%$$
%if $\sigma^2_{i,j}<1$ else $\alpha_{i,j} (k) = 1$.
%\end{itemize}
%%\item Estimation:
%%One has
%%$$
%%p_{i,j} \in []
%%$$
%%with ...
%\item Prediction:
%\begin{itemize}[leftmargin=-0.1cm]
%\item Extinction: 
%\end{itemize}
%\end{enumerate}
%%\textbf{Box1.} PVA method via Bayesian-Galton-Watson approac
%\end{tcolorbox}

Finally, an intermediate approach between the full posterior approach and the punctual estimation consists of building scenarios from the posterior. In PVA, in order to make some decisions, it is (more or less) usual to build some scenarios for the future; see for instance \cite{CFB,CQLC}. Nevertheless, in these papers, the variation of the main estimation to build the difference between the scenarios is not so clear. A natural choice for making scenarios that take into account the data and the error (due to the variance) that we make in the estimation is to consider different quantiles of the posterior. For instance, from Table~\ref{tab:ours}, we can choose for $p_{0,1}(1)$, the values $0.67$, $0.82$, $0.93$ for scenarios that underestimate, estimate or overestimate the mortality rate (note that there is no symmetry).

%\begin{tcolorbox}[colframe=white, colback=yellow!20]
%\textbf{Box 3.} Pathways to make scenario from quantile of the posterior
%\hrulefill
%%\setlist[itemize]{leftmargin=*}
%\begin{enumerate}[leftmargin=*]
%\item  Previous knowledge:
%\begin{itemize}[leftmargin=-0.1cm]
%\item If you have no information, set 
%$$
%\alpha_{i,j} (k) = 1.
%$$
%\item If you expect that $p_{i,j} = m_{i,j}(k) \pm \epsilon_{i,j}(k)$ and $\epsilon_{i,j}(k)$ is centred and has variance of magnitude $$\sigma^2_{i,j}  m_{i,j}(k) (1-m_{i,j}(k))$$ then set
%$$
%\alpha_{i,j} (k) = \frac{(1 - \sigma^2_{i,j}) m_{i,j}(k)}{\sigma^2_{i,j} \sum_{k=0}^\kappa m_{i,j}(k)},
%$$
%if $\sigma^2_{i,j}<1$ else $\alpha_{i,j} (k) = 1$.
%\end{itemize}
%%\item Estimation:
%%One has
%%$$
%%p_{i,j} \in []
%%$$
%%with ...
%\item Prediction:
%\begin{itemize}[leftmargin=-0.1cm]
%\item Extinction: 
%\end{itemize}
%\end{enumerate}
%%\textbf{Box1.} PVA method via Bayesian-Galton-Watson approac
%\end{tcolorbox}

%All of these approaches will be described in the next two sections through concrete examples.

%\subsection{Results}

%\subsection{Approach with scenario}

\section{Examples}
\label{sect:exemple}

In this section, we give two examples. The first one based on simulated data aims to show in a pedagogical way our method and compare it with previous study when population goes into extinction. The second example is a real life example. It is concerned with the French Pyrenean brown bear whose viability has a major interest from an ecological and political point of view.

\subsection*{Synthetic data: comparison of methods}
Let us consider a simple population model going to extinction. We assume there is only one type $K=1$ and the maximum number of offspring is $\kappa_{1,1}=3$. We assume that survival and reproduction events are independent. All individuals survive with probability $p_S=0.4$ and have $k$ offspring with probability $p_R(k)$ with $p_R$ given by the vector
$$
p_R=(0.8, 0.1, 0.05,0.05).
$$
Namely each individual gives birth to $1$ child with probability $0.1$, $2$ with probability $0.05$ and $3$ with probability $0.05$. This is a Galton-Watson process with non-random parameters. With our previous notation, we have
\begin{align*}
&p_{1,1}^\star(0) = (1-p_S)p_R(0)= 0.48\\
&p_{1,1}^\star(1)=p_S p_R(0) + (1-p_S) p_R(1)= 0.38\\
&p_{1,1}^\star(2)=p_S p_R(1) + (1-p_S) p_R(2)= 0.07\\
&p_{1,1}^\star(3)=p_S p_R(2) + (1-p_S) p_R(3)= 0.05\\
&p_{1,1}^\star(4)=p_S p_R(3) = 0.02.
\end{align*}
The expected \textit{mean} matrix $M^\star= \sum_{k=0}^4 k p^\star_{1,1}(k)=0.75$ is just a number and $\lambda^\star=M^\star=0.75<1$ and then extinction is certain. We note $M^\star, \lambda^\star$, $p_{1,1}^\star$, $etc.$ for the true parameters instead of $M,\lambda, p_{1,1},$ $etc.$ because we keep this notation for the random variables distributed according to our prior which will serve to estimate $M^\star,\lambda^\star$, $p_{1,1}^\star$, $etc.$ Let us begin by simulating a population for $5$ years to create a learning data set, that is our knowledge on the population. This is resumed in Table~\ref{tab:data-sim}. Focus on this table, we recover that
$$
N(t) =\sum_{k=0}^4 n_{1,1}(k,t),\  N(t+1)= \sum_{k=0}^4 k n_{1,1}(k,t).
$$
For instance, at time $0$ there are $100$ individuals of whom $47$ die without offspring, $39$ survive without offspring or die with one offspring $etc.$ There are then $75$ individuals at time $1$ and so on.

Continuing the simulation we see that the population evolves as follows
$$
N(T+1)=19, \ N(T+2)=N(T+3)=13,
$$
$$
 N(T+4)=8, \ N(T+5)=5, 
$$
$$
\ N(T+6)=1,\ N(T+7)=0.
$$

Therefore $T_{\text{Ext}}^\star=7$ in this sample. The aim of what follows is to give bounds on this quantity using only Table~\ref{tab:data-sim}.

We choose a non-informative prior $\alpha_{1,1}=1$. We then find that, conditionally on data of Table~\ref{tab:data-sim}, $p_{1,1}$ is distributed according to a Dirichlet law with parameter $(145, 128,20,14,8)$. Indeed it is the sum of each line of Table~\ref{tab:data-sim} plus $\alpha_{1,1}=1$. We then find
$$
\mathbb{E}[M\ | \ \mathbf{n}] = 0.7689, \quad \mathbb{P}(\lambda >1 \ | \ \mathbf{n} ) = 10^{-04}.
$$
That is we find that the posterior mean estimator is relatively close to the true value $M^\star=0.75$ and that the population is not viable with a probability very close to $1$.

Now using the method developed in Appendix~\ref{sect:box1expanded} with thresholds $\alpha=0.05$  we find  that the extinction time  should be between $3$ and $31$; which is what happens because $T_{\text{Ext}}=7$.

We compare our results with two other methods. 

The first one is based on a diffusion model \cite{F94}. Using the method and notation of Foley \cite{F94}, we find that the mean and variance of the growth rate are respectively $r_d=-0.3028$ and $v_r=0.0041$ in log scale. We have $e^{r_d}=0.7387$ which is relatively close to $M^\star=0.75$. We then obtain a good parameter estimation. However to estimate the time to extinction, we have to choose a carrying capacity charge. At best, we obtain $36548$ which is far from the real extinction time. This could be expected because we are not in the setting of \cite{F94} which is concerned with a more stable population (not decreasing one) whose extinction arises with rare events (designed as environmental stochasticity by Foley). 

We also use a naive approach based on linear regression. Starting from Equation~\eqref{eq:K1}, we can fit the trend line:
$$
\log(N(t+1)) = r_d + \log(N_t)= r_d t + c_d,
$$
and then look at the time spent given by the upper and lower confidence lines (at $90\%$ threshold)  to reach $0$. The  lower bound is $7$ and the upper bound is $9$. These bounds are correct because they include the true value $T_{\text{Ext}}=7$.

Taking only one simulation has only a qualitative or pedagogical interest. Let us repeat the previous steps $1000$ times to see how many times the extinction is well predicted. With our method the extinction time is in our interval in around $93\%$ of simulations instead of $49\%$ from the naive regression. To end this example, we plot the density of the extinction through these $1000$ simulations in Figure~\ref{fig:ext-sim} and Figure~\ref{fig:histo}. In  Figure~\ref{fig:ext-sim}, we add the normal law with the same first two moments. We see that the distribution of the extinction time is not Gaussian at all: it is not symmetric and has a heavier tail. It is therefore not suitable to search for some values and build a confidence interval through Gaussian quantiles. In Figure~\ref{fig:histo}, we can illustrate the law of our lower and upper bound.

\subsection*{Real data: French Pyrenean brown bears}
 
 The Pyrenean brown bear (Ursus arctos) population is considered as one of the most seriously threatened with extinction in Western Europe. In the 90's and early 2000s, the reinforcement of the population by the introduction of a few individuals from Slovenia allowed to avoid population extinction. Yet, the population remains fragile, and the introduction of new individuals gives rise to tumultuous debates in France.    
  Let us give here a short study on the powerful properties of our approach (and the difference with the preceding example). %[A detailed study is being prepared?].
 
The data we used comes from \cite{OURS16}. It corresponds to the exhaustive supervision of all the population between $2005$ and $2016$ of the French Pyrenean brown bears. This population is  split into two different isolated subpopulations. In the western part of the Pyrenees, before 2018, there were only two males (of which only one is indigenous), therefore we will focus on the subpopulation living in the central part of the Pyrenees.

We consider the same structure model as \cite{CQLC}. Namely, we only consider the evolution of the number of females and we consider a population structured by age, with $K=5$ classes. For classes $i \in \{1,\dots, 4\},$ this corresponds to bears whose age is $i-1$ years. For $i=5$ this corresponds to the bears whose age is greater than or equal to $4$ years. The life of a bear is modelled as follows: during its first 4 age stages, it can either die or pass to the next age. When it is in the last fully developed stage, it can die, survive and reproduce. See \cite{CQLC} for biological motivations.

We assume that surviving and reproducing are two independent events even if generally it is assumed that reproduction is possible only when the female parent survive. Both choices are modelization preferences that marginally add a bias to the prediction. It can have an importance mostly for the interpretation of the estimated parameters but we are not doing this here. Indeed, as we estimate our paramaters with this assumption, it decreases the values of reproducing parameters (because we will not see offspring when parents die) but it will increase the offspring predictions (because it is possible to have offspring when parents die). These two facts are balanced because they are assumed for the estimation and for the prediction. In the same way, we assume that females can reproduce each year although it is an exceptional event. This changes nothing because it will naturally divide by $2$ the reproduction rate.

Let us now follow the steps of Box 1.

We start by choosing the information that we want to take into account. Begin by considering the non-informative setting, that is, for all $i \in\{1,...,4\}$ and $k \in\{ 0,1 \}$
$$
\alpha_{i,i+1}(k) = 1,
$$
and for $k\in\{0,...,3\}$
$$
\alpha_{5,1}(k)=1.
$$
Using Formula~\eqref{eq:post-gen}, the posterior law is easily calculable from the data of \cite{OURS16} and is then given in Table \ref{tab:ours}. The Beta distribution can be seen as a particular case of Dirichlet distribution. 

However, we only have 11 years of data and one would think that this is not sufficient to estimate all the parameters we need. Let us show how we can integrate the statistical estimation of \cite{CQLC}. The authors used data coming from different articles  with not necessarily the full description of the estimation procedure. If we look at the infantile mortality rate $p_{1,2}(0)$ then it was based on \cite{WNSF}. Their estimation is $0.4$ for a population of $150$ bear cubs. The resulting estimated variance is then $1.96 \sqrt{0.4 \times 0.6/150}= 0.0784$. To take into account this study, we can then choose
$$
\alpha_{1,2}(0)=1.2417, \qquad \alpha_{1,2}(1)=0.7855.
$$
Using these hyper parameters and the data, we find that the posterior law is a Beta$(18. 2417,3.7855)$. This is very close to the non-informative estimation Beta$(18,4)$ given in Table~\ref{tab:ours}. Hence, even if the sample size data seems limited,there is enough information to limit the sensitivity of the prior. Note that we can nevertheless force $(\alpha_{i,j})$ to be large enough to have an impact. 

We then restrict ourselves to the non-informative prior because it gives similar results and several incomes are unknown as for instance for choosing $(\alpha_{1,5})$.

We can now see the  evolution of the population over short time spans. To illustrate this, let us represent for each year $t$, the prediction we give from the previous years $2005, \dots, t-1$.  We represent these successive predictions in Figure~\ref{fig:moustache}. Each boxplot represents the distribution of the prediction. The mean estimator given by Formula~\eqref{eq:post-mean} was added (diamond symbol). To compare with the genuine evolution of the population, we add the total number of female bears at each year (red circles). In this figure, we can for instance see the learning step of our algorithm.

Let us now focus on $\lambda$; that is the viability problem. For this model the principal eigenvalue $\lambda$ is then the solution of the equation
\begin{align*}
- &\lambda^5 + \lambda^4 p_{5,5}(1) +\\
& \sum_{k=1}^4 k p_{5,1}(k)) p_{1,2}(1)p_{2,3}(1)p_{3,4}(1)p_{4,5}(1)=0,
\end{align*}
 which has nevertheless no explicit solution (this equation comes from the fact that $\lambda$ is just the root of the characteristic polynomial). However, numerically, using the algorithm given in Appendix~\ref{sect:box1expanded} one can find
 $$
  \mathbb{P}(\lambda \leq 1)= 0.012. %\mathbb{E}[\lambda \ | \ (n)] =1.08,
 $$
Hence, with high probability the process is super-critical and the reproduction capacity seems to guarantee viability. However if the population is too small, the survival is not guaranteed. Again, surviving probabilities $s$ in Equation \eqref{eq:proba-ext} have no explicit solutions but are functions of $p_{i,j}$ and one can use the algorithm of Appendix~\ref{sect:box1expanded} to find that, with the current population size, the probability of extinction can be rounded to $0$. It seems that the size of the population and the reproduction capacity is large enough to guaranty survival without reintroduction in the central part of the Pyrenees. Let us emphasize again that this result is only valid under our main assumptions. That includes no change in the environmental conditions including hunting or other human interference. Note that we also assumed no change in the basic demographic parameters that could be due to inbreeding depression. In fact, the current population is highly inbreed \cite{OURS16}.

However one can still see in Figure~\ref{fig:s-ours} that it is more beneficial to reintroduce older bears than younger ones. In particular, in 2018, two females were reintroduced in the western part of the Pyrenees. Assuming similar evolution, there is around $15\%$ risk of extinction for this sub-population (of this two males and 2 reintroduced females). Additionally, if these females are pregnant and then if they survive and give birth each to one female cub then this risk will be divided by $2$ (around $8\%$). Closely related to theses questions, using Formula~\eqref{eq:proba-ext}, we can calculate the effective population size; namely the number of reproducing individuals we need to have an extinction risk lower than a defined threshold. Here, to have an extinction probability lower than $0.05$ then we need at least $5$ reproducing females.

\section{Extension to different types of dataset}
\label{sect:ext}

To apply the present method, we need to know exactly the sum over the time of the number of individuals passing from type $i$ to $j$ at each time. This type of information is almost as difficult to collect as the complete genealogical tree of each individual. Even if  several such datasets exist and to the best of our knowledge, few or no previous methods exhaustively use complete information. It is not always possible to have access to this type of full dataset \cite{OURS16, CFB, FF10}. We explain some extensions to generalize our approach in this section. This was not done simply to clarify. 

In the synthetic example of Section~\ref{sect:exemple}, we have to know about the whole life of each individual (survival and reproduction) but in general we only know the number of survivors and the number of offspring, then we can write
$$
p_{1,1} =  p_S +p_R,
$$
as seen written in this example. However, suppose that $p_S$ and $p_R$ are both Dirichlet distributed and independent, this is not the same modelling framework and with this modelling choice, we can study survival and reproduction separately. 

% Elaine, bertrand
 
Another generalisation comes from the second example of the French Pyrenean brown bears. In general we do not know the sex of the bear cub before until a few years after their birth. Consequently, we do not know $n_{4,1}(\cdot, t)$ for the last year, $t$. However, one can slightly change the Bayesian model to take into account this missing information. Indeed, one can add some structure and assume that 
$$
p_{4,1}(k)= \sum_{l\geq k} q(l) \dbinom{l}{k} p^k (1-p)^{l-k} ,
$$
where $q$ is the law of the number of spring for a female bear (male and female bear cubs) and $p$ is the probability that a cub is a female. Using this (natural) representation then we keep the conjugation property, and all the posterior remains explicit. For the dataset \cite{OURS16}, this leads, for instance, to $0.583$ as posterior mean estimator for the sex ratio. This is slightly more favourable than the estimation of $0.5$ used in \cite{CQLC}. Using the number of males in the population and supposing that the survival rates are similar can reduce the variance of our estimation; this is a compromise between the data and expertise.

The assumption that we have made, which can be restrictive, is that surviving and reproducing are independent, but for prediction in contrast with parameter interpretation, this is not a real problem. Nevertheless, one can extend our setting to model the necessity of  survival to make reproduction possible. The assumption that all $p_{i,j}$ are independent from each other can be removed. A careful reading of the proof of Formula~\eqref{eq:post-gen} shows that we can extend it for the following setting: for every $k_1,...,k_K$, we replace all $p_{i,j}(k)$ by $p_i(k_1,...,k_K)$ which is the probability that an individual with type $i$ has $k_1$ offspring of type $1$ and $k_2$ offspring of type $2$ $etc.$ It remains to assume that all $p_i$ are independent and keep a Dirichlet law (on $\prod_{j=1}^K \{1,..,\kappa_{i,j}\}$ instead of $\{1,...,\kappa_{i,j}\}$) as prior to have the same type of approach.

Finally, In many examples coming from ecology, the number of offspring may be very large and it is not possible to fix any value $\kappa_{i,j}$ for the maximum numbers of offspring. It is therefore not possible to directly use our approach unless we take a large $\kappa_{i,j}$ (this requires too much information). In order to keep the same idea for our approach,  we can suppose that offspring are given by a Poisson law $\mathcal{P}(m)$ (which arises naturally as $\kappa_{i,j} \to \infty$ as a limit of our model under general assumptions \cite{BHJ}). Also the naturally associated prior will be the Gamma law. It also verifies properties (conjugation, non-informative, expertise, $etc$.) developed in Section~\ref{sect:BGW} and we also have an explicit posterior distribution.

\section{Discussion}
\label{sect:dis}

We have shown how the Galton-Watson process with Bayesian inference is a simple but powerful setting for PVA analysis. On top of giving the same information as the matrix population model, it easily gives precise results for probability of extinction, time to extinction, $etc$.
It is perfectly adapted for small populations and then perfectly completes  previous studies, see \cite{CRT,ELSW,GSJB,HSVW,OM10}. Deterministic models including differential equation and matrix products are well suited to large population (but sometimes more than $10^6$ individuals \cite{CL12}). Stochastic differential equations or hybrid models (including piecewise deterministic models) are more appropriate at a mesoscopic scale; at least for a thousand individuals. Few models are concerned with smaller population size although these populations are more concerned by extinction.

For conservation purposes, there is more and more population monitoring and there is more and more exhaustive information to understand population evolution. Nevertheless, classical statistical methods do not use this exhaustive information to calibrate demographic model. The Bayesian setting we present offers the advantage of taking into account all the information of these types of dataset and can be adapted for missing data or less informative data.

The combination between Galton-Watson processes and Bayesian, and known mathematical results \cite{AN04,HJV,H02,M16} implies simple formulas and there is no need for large simulations to answer simple questions. Results are found without using long computational algorithms with unknown errors. Consequently, it is easily scalable to big data: in contrast with ABC methods \cite{ELMGPRC} or other MCMC methods \cite{S18}, there is no need for a large number of simulations; this permits us to use this method for large datasets. One can easily use this model as part of a larger model to motivate control of the population or others strategies. Furthermore, there is no doubt about the convergence of this type of algorithm and the associated errors due to such computations.

For small data sets, our approach enables to add expert knowledge and gives precise confidence intervals or error estimation for the prediction. This takes into account statistical errors and the demographic stochasticity. However, the environmental stochasticity is not directly included in this estimation even if the Bayesian learning picks up some part of it with the data.
 
This model is founded on several assumptions and in general, some of them can be questionable such as the absence of carrying capacity, the independence of individuals and the static environment but they are very natural considering a small population during short time period. However, genetic aspects that can be of importance for small populations is not taken into account. Considering the example of the bears, we do not speak about the male reintroduction which can have a real interest to avoid inbreeding depression. We focus on demographic aspects and numbers.
Including the inbreeding issue in our setting may be feasible and could provide interesting insights \cite{D68,H02}. 

\end{multicols}
\begin{center}
\begin{table}
\begin{tabular}{|l|c|c|c|c|c|c|}
  \hline
  Times $t$ & 0 & 1 & 2 & 3 & 4 & 5 \\
  \hline
  $N$ & 100 & 75 & 59 & 43 & 33 & 22\\
 $n_{1,1}(0,t)$ & 47 & 30 & 29 & 20 & 18 & \ \\
  $n_{1,1}(1,t)$ & 39 & 37 & 23 & 17 & 11 & \ \\
 $n_{1,1}(2,t)$ & 8 & 4 & 2 & 3 & 2 & \ \\
 $n_{1,1}(3,t)$ & 4 & 2 & 4 & 2 & 1 & \ \\
  $n_{1,1}(4,t)$ & 2 & 2 & 1 & 1 & 1 & \ \\
  \hline
\end{tabular}
\caption{Synthetic data of Section~\ref{sect:exemple}, $N(t)$ represents the number of individuals at time $t$ and $\mathbf{n} = (n_{1,1}(k,t))_{0\leq k \leq 4, 0 \leq t \leq 4}$ the learning data. } 
\label{tab:data-sim}
\end{table}
\end{center}

\begin{center}
\begin{figure}
\includegraphics[scale=0.4]{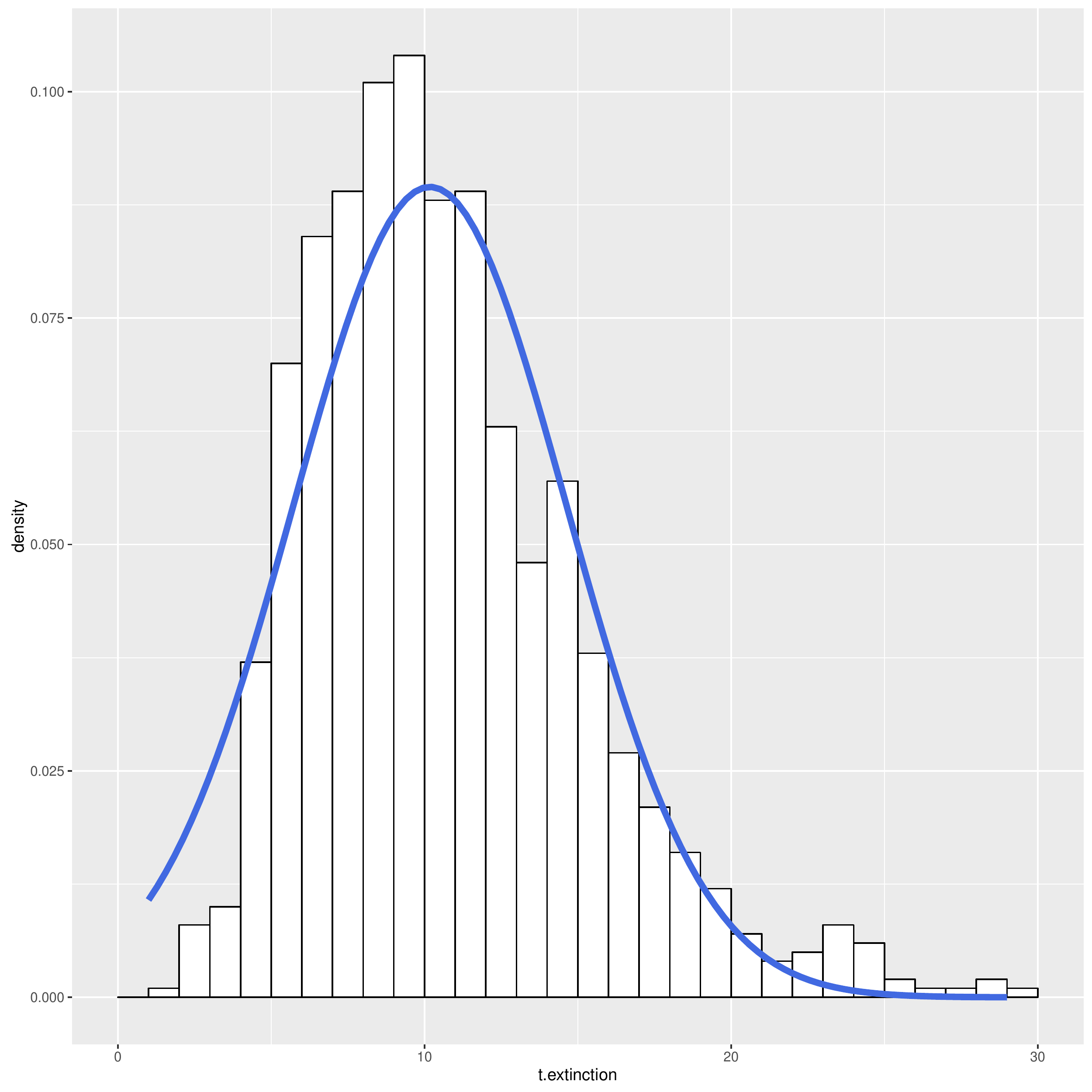}
\caption{The histogram represents the law of the extinction time for the simulated example. The blue curve is the normal distribution with same mean and variance.
}
\label{fig:ext-sim}
\end{figure}
\end{center}

\begin{center}
\begin{figure}
\includegraphics[scale=1]{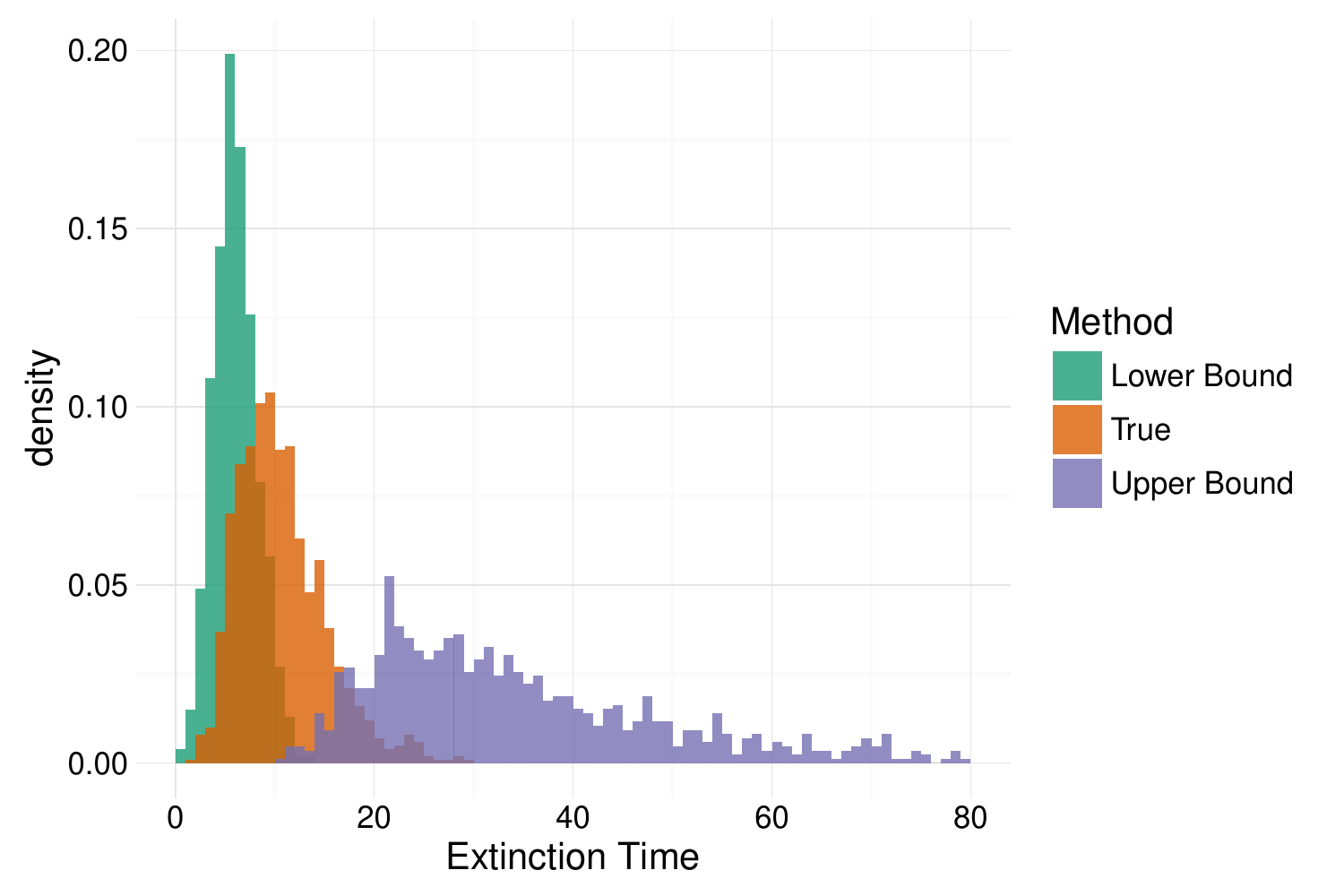}
\caption{The orange histogram represents the law of the extinction time. The other two histograms represent the law of our lower and upper bound. One can see that they intertwine with the beginning and the end of the middle histogram.
}
\label{fig:histo}
\end{figure}
\end{center}

\begin{center}
\begin{table}
\begin{tabular}{|l|c|c|c|}
  \hline
  Parameter & Posterior & Mean & IC 90\%  \\
  \hline
%  Pop. tot. & 128070000 & 127833000 & 127629000 & 127445000 & 127276000 & 127141000\\
 $p_{0,1}(1)$ & Beta (18,4) & 0.8181 & [0.67,0.93] \\
 $p_{1,2}(1)$  & Beta(17,1) & 0.9444 & [0.84,1] \\
 $p_{2,3}(1)$ & Beta(13,3) & 0.8125 & [0.64,0.94] \\
 $p_{3,4}(1)$  & Beta(13,1) & 0.9286 & [0.79,1]\\
 $p_{4,4}(1)$ & Beta(73,4) & 0.9480 & [0.90,0.98] \\
 $p_{4,0}$ & Dir(71,9,7,1) & \ & \ \\ 
 $p_{4,0}(0)$ & Beta(71,17) & 0.8068 & [0.73,0.87]  \\
 $p_{4,0}(1)$ & Beta(9,79) & 0.1023 & [0.05,0.16] \\
 $p_{4,0}(2)$ & Beta(7,81) & 0.0795 & [0.04,0.13]\\
 $p_{4,0}(3)$ & Beta(1,87) & 0.0114 & [0,0.03] \\
  \hline
\end{tabular} 
\caption{Posterior for all parameters $p_{i,j}(k)$ and the vector valued $p_{4,0}$ for the bears population example.  We use the data \cite{OURS16} and the non-informative prior.} 
\label{tab:ours}
\end{table}
\end{center}

\begin{center}
\begin{figure}
\includegraphics[scale=0.4]{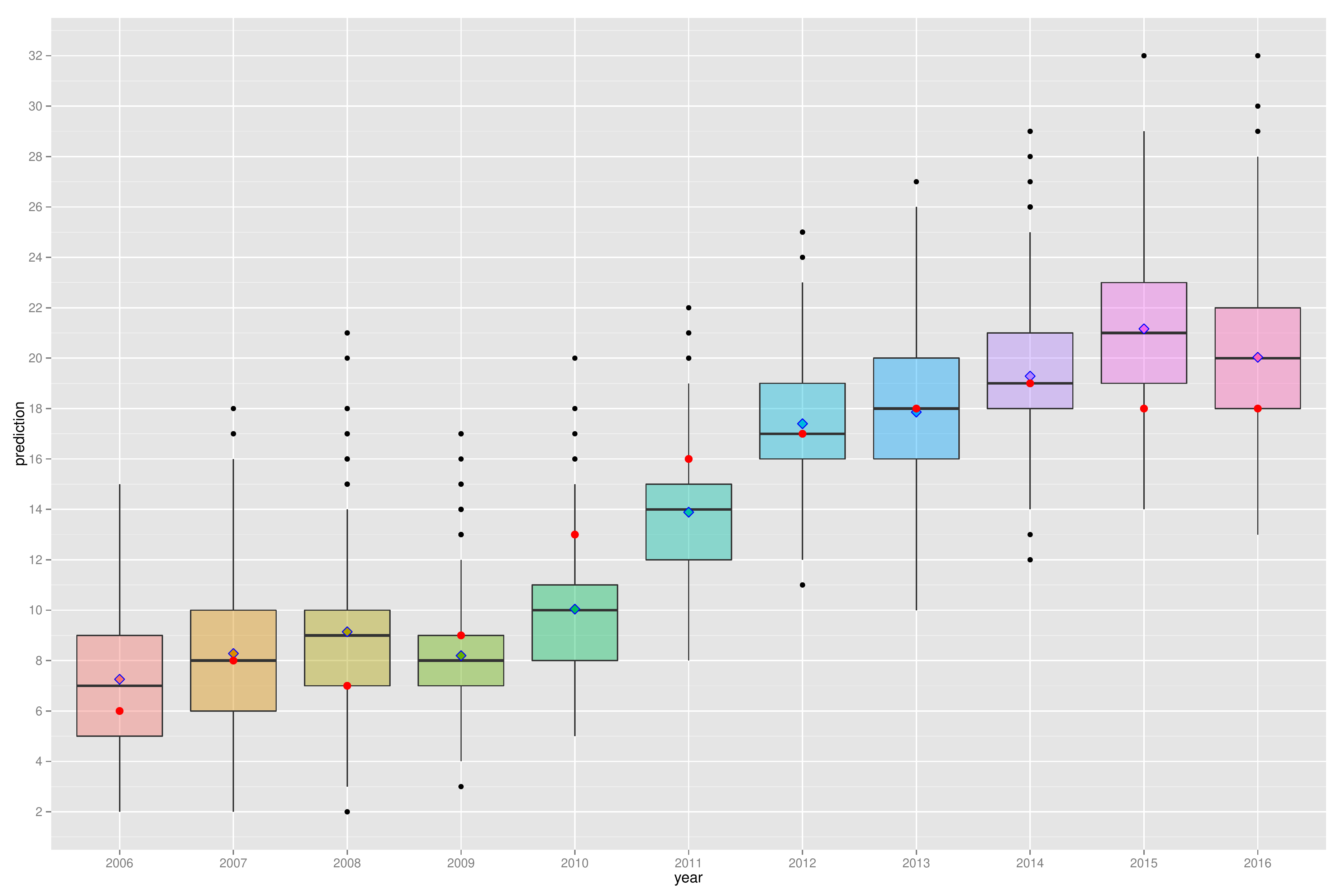}
\caption{Red circles represent the total number of female bears at each year. Boxplots and diamonds both represent prediction of the number pf females at each year considering data available on the previous years and a non-informative prior. The boxplot represents the law of the prediction and diamond represents the mean estimator given by Formula~\eqref{eq:post-mean}.
}
\label{fig:moustache}
\end{figure}
\end{center}

\iffalse
\begin{center}
\begin{table}
\begin{tabular}{|l|c|c|c|c|c|c|c|c|c|c|}
  \hline
  Parameter & $p_{0,1}(1)$ &  $p_{1,2}(1)$  & $p_{2,3}(1)$ &  $p_{3,4}(1)$  &  $p_{4,4}(1)$ & $p_{4,0}$ & $p_{4,0}(0)$ & $p_{4,0}(1)$ &  $p_{4,0}(2)$ &  $p_{4,0}(3)$ \\
  Posterior & Beta (18,4) & Beta(17,1) & Beta(13,3) & Beta(13,1) & Beta(73,4) &  Dir(69,9,9,1) &   Beta(69,19) &   Beta(9,79) & Beta(9,79) &  Beta(1,87) \\
   Mean &  0.8181 & 0.9444 &    0.8125 & 0.9286 &  0.9480 & \ & 0.7841 & 0.1023 &  0.1023 & 0.114 \\
  IC 90\%  & [0.67,0.93] &   [0.84,1] & [0.64,0.94]& [0.79,1]& [0.90,0.98]& \ & [0.71;0.85] & [0.05,0.16] & [0.05,0.16]& [0,0.03] \\
  \hline
\end{tabular} 
\caption{Estimated parameters for ours population} 
\label{tab:ours}
\end{table}
\end{center}

\begin{center}
\begin{table}
\begin{tabular}{|l|c|c|}
  \hline
  Parameter & Value from $\hat{p_{i,j}}$ & IC $10\%$    \\
  \hline
 $s_0$ & 0.59 & [0.41,0.83]  \\
 $s_1$  & 0.5 & [0.30,0.78] \\
 $s_2$ & 0.48 & [0.27,0.77]  \\
 $s_3$  &0.36 & [0.16,0.7]\\
 $s_4$ & 0.31 & [0.12,0.67] \\
  \hline
\end{tabular} 
\caption{Surviving parameters for ours population} 
\label{tab:s-ours}
\end{table}

\end{center}
\fi

\begin{center}
\begin{figure}
\includegraphics[scale=0.3]{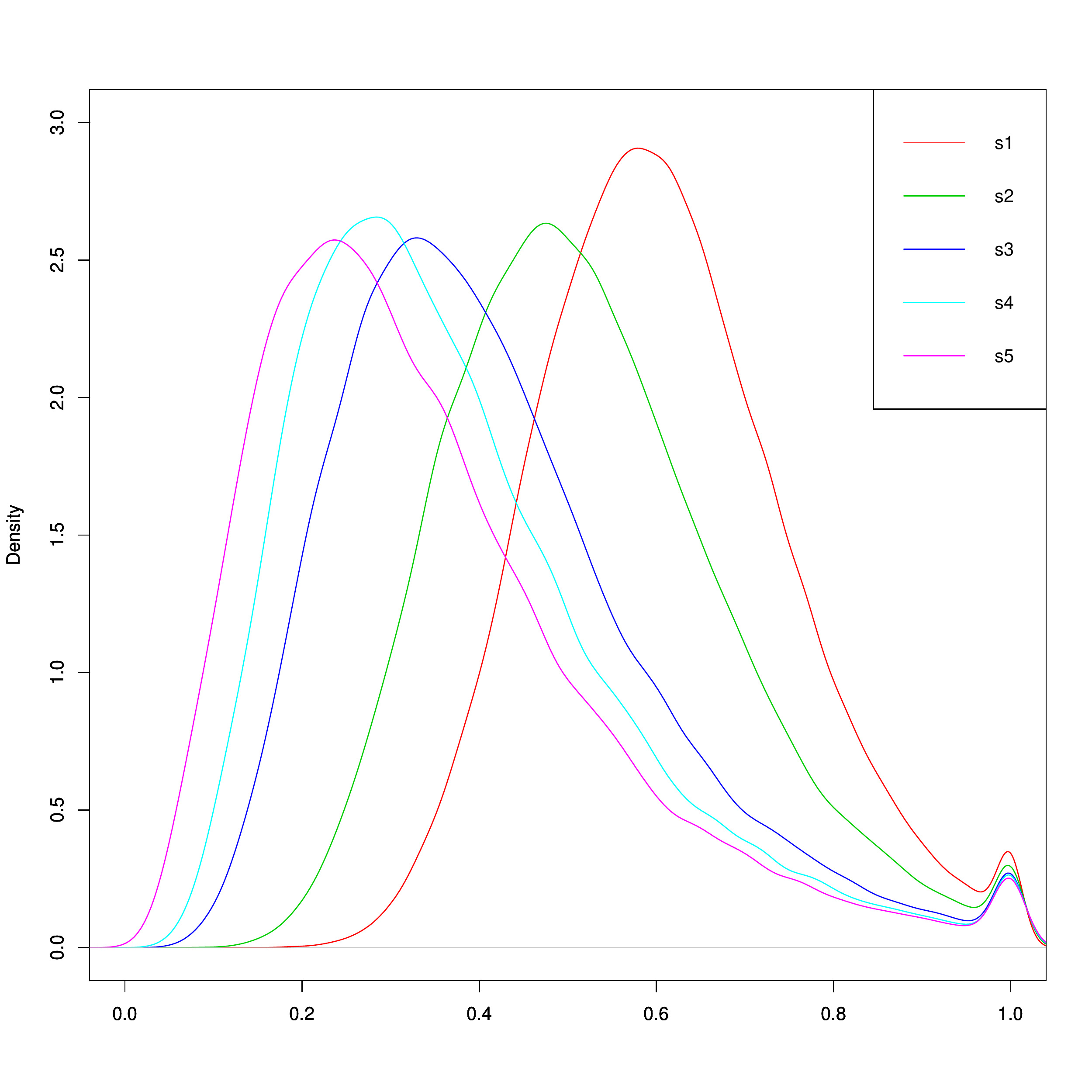}
\caption{Density of parameter $s_1,...,s_5$ for the bear example of Section~\ref{sect:exemple}. Data comes from \cite{OURS16} and we use a non-informative prior.}
\label{fig:s-ours}
\end{figure}
\end{center}

%The drawback is that almost all the tree is needed to predict extinction.

\section{Appendix}
\subsection{Proof of the expression of the Dirichlet distribution}
\label{sect:post}
First  we recall that
$$
p_{i,j} \sim \mathbf{Dir}\left(\alpha_{i,j}(0),..., \alpha_{i,j}(\kappa_{i,j})\right),
$$ 
means that for all $a_1,...a_{\kappa_{i,j}}, b_1,...,b_{\kappa_{i,j}}$, setting $A=[a_1,b_1] \times \cdots \times [a_{\kappa_{i,j}},b_{\kappa_{i,j}}]$, we have
\begin{align*}
\mathbb{P}&(p_{i,j} \in A)\\
&=\int_A \frac{\Gamma (\sum_{k=1}^{\kappa_{i,j}} \alpha_{i,j}(k))}{\prod_{k=1}^{\kappa_{i,j}} \Gamma (\alpha_{i,j}(k))} \\
&\quad \times \prod_{k=1}^{\kappa_{i,j}} x_k^{\alpha_{i,j}(k)-1} dx_1 \dots dx_{\kappa_{i,j}},
\end{align*}
where $\Gamma$ is the classical Gamma function. Then, by definition of the model we have the following prior:
$$
\mathcal{L}(p_{i,j})= \mathbf{Dir}\left(\alpha_{i,j}(0),..., \alpha_{i,j}(\kappa_{i,j})\right).
$$
Let us now detail the likelihood of the model using multinomial laws. Let us denote by $N_{i,j}(t,k)$  the number of individuals with type $i$ having $j$ descendants of type $i$ at time $t$. Namely
$$
N_{i,j}(t,k) = \sum_{l=1}^{N_i(t)} \mathbf{1}_{ \{ \xi_{i,j,l,t}= k \}}.
$$
We have for all $t\geq 0$
\begin{align*}
&\mathcal{L} \left( \left(N_{i,j} (k,t) \right)_{0\leq k  \leq \kappa_{i,j}} \ | \ (p_{i,j}(k))_{0 \leq k \leq \kappa_{i,j} }, N_i(t)  \right) \\
= & \mathbf{M} ( N_i(t), (p_{i,j}(k))_{0 \leq k \leq \kappa_{i,j} } ),
\end{align*}
where $\mathbf{M}(n, q_0,...,q_\kappa)$ denotes the usual multinomial law with $n$ trials and $q_0,...,q_m$ events probabilities. Namely $M\sim \mathbf{M}(n, q_0,...,q_\kappa)$ means that for all $m_0, \dots, m_\kappa$ we have
\begin{align*}
&\mathbb{P} \left( M(0)=m_0, \dots, M(\kappa) =m_\kappa  \right) \\
= & \frac{n !}{m_0! \dots m_\kappa! } q_0^{m_0} \times \dots \times q_\kappa^{m_\kappa},
\end{align*}
when $\sum_{i=0}^\kappa m_i=n$ either it is equal to $0$.
Using these two explicit formulas (multinomial and Dirichlet densities) and the Bayes Theorem, it is easy to check the classical conjugation result:
\begin{align*}
&\mathcal{L} \left( (p_{i,j}(k))_{1 \leq k \leq \kappa_{i,j} },  \ | \  \left(N_{i,j} (k,t) \right)_{1\leq k  \leq \kappa_{i,j}}, N_i(t)  \right) \\
= & \mathbf{Dir}\left(\alpha_{i,j}(0) + N_{i,j} (0,t) ,..., \alpha_{i,j}(\kappa_{i,j}) + N_{i,j} (\kappa_{i,j},t)\right).
\end{align*}
Further details can be found in \cite{R07,R09}. It remains to prove that it holds true for any times through a chain rule type argument. The Markov property (derived by the independence assumption) gives that
\begin{align*}
&\mathcal{L} \left( \left(N_{i,j} (k,t) \right)_{0\leq k  \leq \kappa_{i,j}, 0 \leq t \leq T } \ | \ (p_{i,j}(k))_{0 \leq k \leq \kappa_{i,j} }  \right) \\
= &\mathcal{L} \left( \left(N_{i,j} (k,t) \right)_{0\leq k  \leq \kappa_{i,j}, 0 \leq t \leq T } \ | \ \left(N_{i,j} (k,t) \right)_{0\leq k  \leq \kappa_{i,j} }, (p_{i,j}(k))_{0 \leq k \leq \kappa_{i,j} }  \right) \\
\otimes & \ \mathcal{L} \left( \left(N_{i,j} (k,t) \right)_{0\leq k  \leq \kappa_{i,j}, 0 \leq t \leq T-1 } \ | \ (p_{i,j}(k))_{0 \leq k \leq \kappa_{i,j} }  \right) \\
= &\mathcal{L} \left( \left(N_{i,j} (k,T) \right)_{0\leq k  \leq \kappa_{i,j}} \ | \ N_i(T), (p_{i,j}(k))_{0 \leq k \leq \kappa_{i,j} }  \right) \\
\otimes & \ \mathcal{L} \left( \left(N_{i,j} (k,t) \right)_{0\leq k  \leq \kappa_{i,j}, 0\leq t \leq T-1 } \ | \ (p_{i,j}(k))_{0 \leq k \leq \kappa_{i,j} }  \right) \\
= & \prod_{t=1}^T \mathcal{L} \left( \left(N_{i,j} (k,t) \right)_{0\leq k  \leq \kappa_{i,j}, 0 } \ | \ N_i(t), (p_{i,j}(k))_{0 \leq k \leq \kappa_{i,j} }  \right).
\end{align*}
Using this formula, the expressions of multinomial and Dirichlet densities and Bayes Theorem again finalise the proof of Equation~\eqref{eq:post-gen}. Note that $\mathbf{n}$ that represents the data is a realisation of the random variables $\left(N_{i,j} (k,t) \right)_{0\leq k  \leq \kappa_{i,j}, 0\leq t \leq T }$.

\subsection{Pathways to PVA through the Bayesian-Galton-Watson approach: the expanded version }
\label{sect:box1expanded}
\

Here we give details for each step of Box 1.

\textbf{Previous knowledge} \\
Initializes $(\alpha_{i,j}(k))$ as follows:
\begin{enumerate}
\item[i)] If one cannot pass from state $i$ to state $j$, set
$$
p_{i,j}(0)=1.
$$
or equivalently, in a certain sense,
$$
\alpha_{i,j} (k)=\mathbf{1}_{k=0}.
$$

\item[ii)] If you have no information, set 
$$
\alpha_{i,j} (k) = 1.
$$
\item[iii)] If we use previous estimates and we expect that $p_{i,j} = m_{i,j}(k) \pm \epsilon_{i,j}(k)$ and $\epsilon_{i,j}(k)$ is centred and has variance of magnitude $$\sigma^2_{i,j}  m_{i,j}(k) (1-m_{i,j}(k))$$ then set
$$
\alpha_{i,j} (k) = \frac{(1 - \sigma^2_{i,j}) m_{i,j}(k)}{\sigma^2_{i,j} },
$$
if $\sigma^2_{i,j}<1$ else $\alpha_{i,j} (k) = 1$.
\item[iv)] If we want to impose an expert opinion. The experts think that 
$$
p_{i,j} \approx q_{i,j}
$$
and that their opinion is as important as $M$ data then we can choose
$$
\alpha_{i,j} = M \times q_{i,j}.
$$
\end{enumerate}
We can see in case iii), that this choice enables to match the expectation to the variance. More precisely, to match the expectation, we need 
$$
\frac{\alpha_{i,j}(k)}{S_{i,j}} =m_{i,j}(k)
$$
with $S_{i,j}=\sum_{k=0}^{\kappa_{i,j}} \alpha_{i,j}(k)$. We have then $\alpha_{i,j}(k) =S_{i,j} m_{i,j}(k)$. To match the variance, we have
\begin{align*}
&\frac{S_{i,j} m_{i,j}(k) (S_{i,j}-S_{i,j}m_{i,j}(k)}{S_{i,j}^2(S_{i,j}+1)} \\
= &m_{i,j}(k)(1-m_{i,j}(k)) \sigma_{i,j}^2,
\end{align*}
hence 
$$
S_{i,j}= \frac{1}{\sigma_{i,j}^2} -1 =  \frac{1 - \sigma_{i,j}^2}{\sigma_{i,j}^2}.
$$
Of course we need  $\sigma_{i,j}^2<1$. However, if this is not the case, the noise is too large and the information on $p_{i,j}$ is therefore less informative than no information.

\textbf{Calculus of quantities of interest}

For $K$ ($K=1$ or $2$), one can directly give an expression for almost all quantities through some integral calculus. However, in general, it is difficult to express them. The better way to calculate this quantity is perhaps the following simple Monte Carlo algorithms.

Due to the central limit theorem the error of this simple algorithm can be estimated through an estimator of the variance but in general one cannot fix a number of simulations before doing the calculation. However for the probability of viability and the probability of extinction, the variance can be bounded by $1/4$ and then the error of this simple algorithm is lower than $(4 \sqrt{n_{\text{prec}}})^{-1}$. For an error lower than $0.5\%$ we can take $n_{\text{prec}}=2500$.

Below, simulate $p_{i,j}$ means drawing a random variable with law
$$
\mathbf{Dir}\left(\alpha_{i,j}(0) + \sum_{t=0}^T n_{i,j}(0,t), \dots, \alpha_{i,j}(\kappa_{i,j}) + \sum_{t=0}^Tn_{i,j}(\kappa_{i,j},t)\right).
$$

\textbf{Short time evolution}: 

Do $n_{\text{prec}}$ times the steps
 
\begin{itemize}
\item Simulate the $K^2$ probability valued random variables $p_{i,j}$.
\item Calculate $M^n \cdot X_0$ with $M$ defined by \eqref{eq:M}.
\end{itemize}
Then do a mean value of these proportions.

\textbf{Probability of viability}

Do $n_{\text{prec}}$ times the steps
 
\begin{itemize}
\item Simulate the $K^2$ probability valued random variables $p_{i,j}$.
\item Take for $\lambda$ the largest eigenvalue of the $M$ defined by \eqref{eq:M}.
\end{itemize}
Then count the proportion of $\lambda$ being larger than $1$.

\textbf{Time to extinction}:
To bound the extinction time $T_{\text{Ext}}$ between two values
$$
T_- \leq T_{\text{Ext}} \leq T_+,
$$
with probability $1-\alpha$, for some fixed threshold $\alpha$, do $n_{\text{prec}}$ times the steps
\begin{itemize}
\item Simulate the $K^2$ probability valued random variables $p_{i,j}$.
\item Calculate the function
$$
U(t)= \lambda^t \frac{\sum_{j=1}^K N_j(T) v_j}{ \min(v)}
$$
and 
$$
L(t)= \frac{\max(v)^2}{\min(v)^2} \frac{(1-\lambda)}{\Xi} \lambda^{t+1}  \sum_{j=1}^K N_j(T) v_j,
$$
where
$$
\Xi=\sum_{j=1}^{K} \frac{v_j^2}{\min(v)} \sup_{1\leq i \leq K} \sum_{k=1}^{\kappa_{i,j}}(k^2-M_{i,j}^2) p_{i,j}(k)
$$
 where $v$ is given in Equation~\eqref{eq:vp}.
\end{itemize}
Then do a mean value of functions $U$ and $L$ and choose $(T_-, T_+)$ such that
$$
L(T_-)=\alpha, \qquad U(T_+)= \alpha.
$$
These two functions are based on Equation~\eqref{eq:tps-ext} and Lemma~\ref{lem:tps-ext}.
Note that these bounds on extinction time only hold when $\lambda<1$. To avoid some computational problems due to the fact that $\mathbb{P}(\lambda <1)>0$ (even  if it is very small), it can be useful, when calculating the extinction time conditioned on the event $\{\lambda <1 \}$. In the previous algorithm, this conditioning translates into doing a mean value only for $\lambda$ satisfying $\lambda <1$.

\textbf{Probability of extinction}: 
 
Do $n_{\text{prec}}$ times the steps
 
\begin{itemize}
\item Simulate the $K^2$ probability valued random variables $p_{i,j}$.
\item Find the (vector) solution $s$ of equation $\varphi(s)=s$, with $\varphi$ is the generating function defined in \eqref{eq:phi}. To do this step one can use any classical optimization algorithm.
\item Calculate $s_1^{N_1(T)} \times \cdots \times s_K^{N_K(T)}$.
\end{itemize}
Then do a mean of these quantities.

\textbf{How to plan reintroduction}: 
 Do $n_{\text{prec}}$ times the steps
 
\begin{itemize}
\item Simulate the $K^2$ probability valued random variables $p_{i,j}$.
\item Find the (vector) solution $s$ of equation $\varphi(s)=s$, with $\varphi$ is the generating function defined in \eqref{eq:phi}. To do this step one can use any classical optimization algorithm.
\end{itemize}
Then do a histogram or smoothed density through classical kernel density methods for instance. 

\subsection{Lower bound for the extinction time}

\begin{lem}
\label{lem:tps-ext}
We have
\begin{equation}
\label{eq:lwb}
 \mathbb{P}\left( T_{\text{Ext}} \geq t \right) \geq  \frac{\max(v)^2}{\min(v)^2}   \frac{\lambda^{2t} \left( \sum_{j=1}^K v_j N_j(0) \right)^2}{ \Xi \lambda^{t-1} \left(\frac{1-\lambda^{t}}{1-\lambda} \right) \left( \sum_{j=1}^K v_j N_j(0) \right) + \lambda^{2t} \left( \sum_{j=1}^K v_j N_j(0) \right)^2 },
\end{equation}
where
$$
\Xi=\sum_{j=1}^{K} \frac{v_j^2}{\min(v)} \sup_{1\leq i \leq K} \sum_{k=1}^{\kappa_{i,j}}(k^2-M_{i,j}^2) p_{i,j}(k).
$$
 Moreover, the right-hand side of equation~\eqref{eq:lwb} is equivalent to
$$
\frac{\max(v)^2}{\min(v)^2}   \frac{\lambda^{2t} \left( \sum_{j=1}^K v_j N_j(0) \right)^2}{ \Xi \lambda^{t-1} \left(\frac{1-\lambda^{t}}{1-\lambda} \right) \left( \sum_{j=1}^K v_j N_j(0) \right) + \lambda^{2t} \left( \sum_{j=1}^K v_j N_j(0) \right)^2 }
 \underset{t \to \infty}{\sim} \frac{\max(v)^2}{\min(v)^2}  \lambda^{t+1} \left( \sum_{j=1}^K v_j N_j(0) \right) \frac{(1-\lambda)}{\Xi}.
$$
\end{lem}
\begin{proof}
The proof is inspired from \cite[Box 5.2 p.119]{HJV} We have
\begin{align*}
 \mathbb{P}\left( \sum_{j=1}^K N_j(t) \geq 1 \right)
% &\geq  \mathbb{P}\left( \sum_{j=1}^K v_j N_j(t) \geq \max(v) \right)\\
  &\geq \frac{\max(v)^2}{\min(v)^2}  \frac{\mathbb{E}[X(t)]^2}{\mathbb{E}[X(t)^2]},
\end{align*}
where
$$
X(t)= \sum_{j=1}^K v_j N_j(t) .
$$
Yet, we have $\mathbb{E}[X(t)] = \lambda^t \sum_{j=1}^{K} v_j N_j(0)$, and it remains to study the second moment. Recall 
$$
N_j(t+1) = \sum_{i=1}^K \sum_{l=1}^{N_{i}(t)} \xi_{i,j,l,t}.
$$
Setting $\mathbf{N}(t) = \sum_j N_j(t)$ for the population size and reindexing the population, we find
$$
N_j(t+1) = \sum_{l=1}^{\mathbf{N}(t)} \zeta_{i(l),j,l,t},
$$
where $i(l)$ is the type of the individual $l$ and $\zeta$ its offspring. We have
\begin{align*}
\mathbb{E}[X(t+1)^2] &= \mathbb{E} \left[\left( \sum_{l=1}^{\mathbf{N}(t)} \sum_{j=1}^{K} v_j \zeta_{i(l),j,l,t} \right)^2\right]\\
 &= \mathbb{E} \left[ \sum_{l,l'=1}^{\mathbf{N}(t)} \sum_{j,j'=1}^{K} \left(  v_j \zeta_{i(l),j,l,t} \right)\left(  v_{j'} \zeta_{i(l'),j',l',t} \right)\right]\\
  &= \mathbb{E} \left[ \sum_{l,l'=1}^{\mathbf{N}(t)} \sum_{j,j'=1}^{K} \mathbb{E} \left[\left(  v_j \zeta_{i(l),j,l,t} \right)\left(  v_{j'} \zeta_{i(l'),j',l',t} \right]\right])\right]\\
 &= \mathbb{E} \left[ \sum_{l,l'=1}^{\mathbf{N}(t)} \sum_{j,j'=1}^{K} v_j v_{j'} (\mathbf{1}_{l=l'} M^{(2)}_{i(l),j,j'} + \mathbf{1}_{l\neq l'} M_{i(l),j} \times M_{i(l'),j'})\right] \\
  &= \mathbb{E} \left[  \sum_{j,j'=1}^{K} v_j v_{j'} \left( \sum_{i=1}^K N_i(t) M^{(2)}_{i,j,j'} + \sum_{i,i'=1}^K  N_i(t)  N_{i'}(t) M_{i,j} \times M_{i',j'}- \sum_{i=1}^K  N_i(t) M_{i,j} \times M_{i',j'}\right)\right].
\end{align*}
In the third line we conditioned on $\mathbf{N}(t)$  and use the independance property. In the fourth line, we use the notation
$$
M^{(2)}_{i(l),j,j'} = \mathbb{E} \left[\left(  v_j \zeta_{i(l),j,l,t} \right)\left(  v_{j'} \zeta_{i(l'),j',l',t} \right) \right],
$$
and then under our assumptions 
$$
M^{(2)}_{i(l),j,j'}= M_{i(l),j} M_{i(l),j'}, 
$$
if $j\neq j'$ and either
$$
M^{(2)}_{i(l),j,j}= \sum_{k=0}^{\kappa_{i(l),j}} k^2 p_{i,j}(k).
$$
Then
\begin{align*}
 &\sum_{j,j'=1}^{K} v_j v_{j'} \left( \sum_{i=1}^K N_i(t) M^{(2)}_{i,j,j'} + \sum_{i,i'=1}^K  N_i(t)  N_{i'}(t) M_{i,j} \times M_{i',j'}- \sum_{i=1}^K  N_i(t)^2 M_{i,j} \times M_{i',j'}\right)\\
 = &\lambda^2 X(t)^2  +\sum_{j,j'=1}^{K} v_j v_{j'} \left( \sum_{i=1}^K N_i(t) M^{(2)}_{i,j,j'} - \sum_{i=1}^K  N_i(t) M_{i,j} \times M_{i',j'}\right)\\
  = &\lambda^2 X(t)^2  +\sum_{j=1}^{K} \sum_{i=1}^K N_i(t) v_j^2 \left(M^{(2)}_{i,j,j'} -  M_{i,j} \times M_{i',j'}\right)\\
    = &\lambda^2 X(t)^2  +\sum_{j=1}^{K} \sum_{i=1}^K N_i(t) v_j^2 \sigma^2_{i,j}.
\end{align*}
and also setting $\overline{\sigma_j^2} =\sup_i\sigma^2_{i,j} $
\begin{align*}
\sum_{j=1}^{K} \sum_{i=1}^K N_i(t) v_j^2 \sigma^2_{i,j}
&\leq  \sum_{i=1}^K N_i(t) v_i \sum_{j=1}^{K} v_j^2 \overline{\sigma_j^2}/\min(v).
\end{align*}
Finally setting $\Xi=\sum_{j=1}^{K} v_j^2 \overline{\sigma_j^2}/\min(v)$, we obtain
\begin{align*}
\mathbb{E}[X(t+1)^2] 
&\leq \lambda^2 \mathbb{E}[X(t)^2] + \Xi  \mathbb{E}[X(t)]\\
&\leq \lambda^2 \mathbb{E}[X(t)^2] + \Xi \lambda^t X(0) \\
&\leq \Xi \lambda^t \left(\frac{1-\lambda^{t+1}}{1-\lambda} \right) X(0) + \lambda^{2(t+1)} X(0)^2 
\end{align*}
where we used an iteration argument in the third line. To conclude

\begin{align*}
 \mathbb{P}\left( \sum_{j=1}^K N_j(t) \geq 1 \right)
 %&\geq  \mathbb{P}\left( \sum_{j=1}^K v_j N_j(t) \geq \max(v) \right)\\
  &\geq \frac{\max(v)^2}{\min(v)^2}   \frac{\lambda^{2t} X(0)^2}{ \Xi \lambda^{t-1} \left(\frac{1-\lambda^{t}}{1-\lambda} \right) X(0) + \lambda^{2t} X(0)^2 }\\
  &\sim \frac{\max(v)^2}{\min(v)^2}  \lambda^{t+1} X(0) \frac{(1-\lambda)}{\Xi}.
\end{align*}

\end{proof}

\bibliographystyle{abbrv}%{amsalpha}%{abbrv} %alpha abbrv
\bibliography{ref}

\end{document}